\documentclass[runningheads]{llncs}
\usepackage{comment}
\usepackage{graphicx}
\usepackage{macros}
\begin{document}
\title{Data-graph repairs: the \textit{preferred} approach
}
%
%
\author{Abriola Sergio\inst{1,2} \and 
Cifuentes Santiago\inst{1,2} \and
Pardal Nina\inst{2,3} \and
Pin Edwin\inst{2}}
\authorrunning{S. Abriola et al.}
%
\institute{Departamento de Computación, Facultad de Ciencias Exactas, Universidad de Buenos Aires, Argentina \and ICC CONICET, Buenos Aires, Argentina \and Department of Computer Science, University of Sheffield, UK}
%
\maketitle              
\begin{abstract}
Repairing inconsistent knowledge bases is a task that has been assessed, with great advances over several decades, from within the knowledge representation and reasoning and the database theory communities. As information becomes more complex and interconnected, new types of repositories, representation languages and semantics are developed in order to be able to query and reason about it. Graph databases provide an effective way to represent relationships among data, and allow processing and querying these connections efficiently. In this work, we focus on the problem of computing preferred (subset and superset) repairs for graph databases with data values, using a notion of consistency based on a set of Reg-GXPath expressions as integrity constraints. Specifically, we study the problem of computing preferred repairs based on two different preference criteria, one based on weights and the other based on multisets, showing that in most cases it is possible to retain the same computational complexity as in the case where no preference criterion is available for exploitation.
\end{abstract}

\section{Introduction}
Graph databases are useful in many modern applications where the topology of the data is as important as the data itself, such as social networks analysis \cite{fan2012graph}, data provenance \cite{anand2010techniques}, and the Semantic Web \cite{arenas2011querying}. The structure of the database is commonly queried through navigational languages such as \textit{regular path queries} or RPQs \cite{barcelo2013querying} that can capture pair of nodes connected by some specific kind of path. This query languages can be extended to add more expressiveness, while usually adding extra complexity in the evaluation as well. For example, C2RPQs  are a natural extension of RPQs defined by adding to the language the capability of traversing edges backwards and closing the expressions under conjunction (similar to relational CQs). 

RPQs and its most common extensions (C2RPQs and NREs \cite{barcelo2012relative}) can only act upon the edges of the graph, leaving behind any possible interaction with data values in the nodes. This led to the design of query languages for \emph{data-graphs} (i.e. graph databases where data lies both in the paths and in the nodes themselves), such as REMs and \Gregxpath \cite{libkin2016querying}.  

As in the relational case, it is common to expect that the data preserves some semantic structure related to the world it represents. These \textit{integrity constraints} can be expressed in graph databases through \textit{path constraints}~\cite{abiteboul1999regular,buneman2000path}.

When a database does not satisfy its integrity constraints, a possible approach is to search for a `similar' database that does satisfy the constraints. In the literature, this new database is called a \emph{repair}~\cite{arenas1999consistent}, and in order to define it properly one has to precisely define the meaning of `similar'.

In the literature one can find different notions of repairs, among others, set-based repairs~\cite{tenCate:2012}, attribute-based repairs~\cite{Wijsen:2003}, and cardinality based repairs~\cite{Lopatenko07complexityof}.
When considering set-based repairs $G'$ of a graph database $G$ under a set of \Gregxpath expressions $R$, two natural restrictions of the problem are when $G'$ is a sub-graph of $G$ and when $G'$ is a super-graph of $G$. These kind of repairs are usually called \textit{subset} and \textit{superset} repairs respectively~\cite{tenCate:2012,barcelo2017data}. Since repairs may not be unique, it is possible to impose an ordering over the set of repairs and look for an `optimum' repair over such ordering.
There is a significant body of work on preferred repairs
for relational databases~\cite{flesca2007preferred,staworko2012prioritized} and other types of logic-based formalisms~\cite{brewka89,bienvenu2014querying}. However, to the best of our knowledge, there is no such work focused on graph databases or data-graphs.
In this work, we study the problem of finding a preferred repair based on two preference criteria that we propose.

This work is organized as follows. In Section~\ref{Section:Definitions} we introduce the necessary preliminaries and notation for the syntax and semantics for our data-graph model as well as the definitions of consistency and different types of repairs.
In Section~\ref{Section:PrefRepairs} we develop two different proposals to assign preferences to repairs. The first one is based on the assignment of weights, and the second one is based on lifting an ordering over edges and data to multiset orderings. For both proposals we study the computational complexity of the problem of computing a preferred repair. Conclusions and future work directions are discussed in Section~\ref{Section:Conclusions}.

\section{Definitions}\label{Section:Definitions}
\vspace{-2mm}
Fix a finite set of edge labels $\Sigma_e$ and a countable (either finite or infinite enumerable) set of data values $\Sigma_n$ (sometimes called data labels), which we assume non-empty and with $\Sigma_e \cap \Sigma_n = \emptyset$. A \defstyle{data-graph}~$G$ is a tuple $(V,L_e,\dataFunction)$ where $V$ is a set of nodes, $L_e$ is a mapping from $V \times V$ to $\parts{\Sigma_e}$ defining the edges of the graph, and $\dataFunction$ is a mapping from $V$ to the set of data values $\Sigma_n$. 

\Gregxpath \textit{expressions} are given by the following mutual recursion: 
\begin{align*}
    \aFormula, \aFormulab &:= \esDatoIgual{\aData} \mid \esDatoDistinto{\aData} \mid \neg \aFormula \mid \aFormula \vee \aFormulab \mid \aFormula \wedge \aFormulab \mid \comparacionCaminos{\aPath} \mid  \comparacionCaminos{\aPath = \aPathb} \mid \comparacionCaminos{\aPath \neq \aPathb}  \\    
    \aPath, \aPathb &:= \epsilon \mid \labelComodin \mid \aLabel \mid \aLabel^{-} \mid \aPath \circ \aPathb \mid \aPath \pathUnion \aPathb \mid \aPath \pathIntersection \aPathb \mid \aPath^{*} \mid \pathComplement{\aPath} \mid \expNodoEnCamino{\aFormula} \mid    \aPath^{n,m}    
\end{align*}
where $\aData$ iterates over $\Sigma_n$ and $\aLabel$ iterates over  $\Sigma_e$. Formulas like $\aFormula$ are called \textit{node expressions} and formulas like $\aPath$ are called \textit{path expressions}. 
The subset of \Gregxpath called \Gcorexpath is obtained by allowing the Kleene star to be applied only to labels and their inverses (i.e. $\aLabel^-$).
The semantics of these languages are defined in \cite{libkin2016querying} in a similar fashion as the usual regular languages for navigating graphs \cite{barcelo2013querying}, also adding some extra capabilities such as the complement of a path expression $\pathComplement{\aPath}$ and data tests. The $\comparacionCaminos{\aPath}$ operator is the usual one for \textit{nested regular expressions} (or NREs) used in \cite{barcelo2012relative}. Given a data-graph $\aGraph=(V,L,D)$, the semantics are:
\begin{alignat*}{2}
&\semantics{\epsilon}_\aGraph = \{ (v,v) \mid v\in V \} &&\hspace{-5.8cm} \semantics{\labelComodin}_\aGraph = \{(v,w) \mid L(v,w) \neq \emptyset\} \\ 
&\semantics{\aLabel}_\aGraph = \{ (v,w) \mid \aLabel \in L(v,w)\}   
&&\hspace{-5.8cm} \semantics{\aLabel^-}_\aGraph = \{ (w,v) \mid  \aLabel \in L(v,w)\}  \\
&\semantics{\aPath^*}_\aGraph = \text{the reflexive transitive closure of } \semantics{\aPath}_\aGraph \\ 
&\semantics{\aPath \star \aPathb}_\aGraph = \semantics{\aPath}_\aGraph \star \semantics{\aPathb}_\aGraph \text{ for } \star \in \{\circ,\cup,\cap\} \\
&\semantics{\pathComplement{\aPath}}_\aGraph = V \times V \setminus \semantics{\aPath}_\aGraph 
&&\hspace{-5.8cm} \semantics{\expNodoEnCamino{\aFormula}}_\aGraph = \{(v,v) \mid v \in \semantics{\aFormula}_\aGraph\} \\
&\semantics{\esDatoIgual{c}}_\aGraph = \{v \in V \mid D(v) = \aData\}  
&&\hspace{-5.8cm}\semantics{\esDatoDistinto{c}}_\aGraph = \{v \in V \mid D(v) \neq \aData\} \\ 
&\semantics{\aFormula \wedge \aFormulab}_\aGraph = \semantics{\aFormula}_\aGraph \cap \semantics{\aFormulab}_\aGraph   
&&\hspace{-5.8cm} \semantics{\aFormula \vee \aFormulab}_\aGraph = \semantics{\aFormula}_\aGraph \cup \semantics{\aFormulab}_\aGraph \\
&\semantics{\lnot \aFormula}_\aGraph = V \setminus \semantics{\aFormula}_\aGraph    
&&\hspace{-5.8cm} 
\semantics{\comparacionCaminos{\aPath}}_\aGraph = \{v \mid \exists w\in V,\, (v,w) \in \semantics{\aPath}_\aGraph \} \\
&\semantics{\comparacionCaminos{\aPath = \aPathb}}_\aGraph = \{v \mid \exists u, w,\, (v,u) \in \semantics{\aPath}_\aGraph,\, (v,w) \in \semantics{\aPathb}_\aGraph,\, D(u) = D(w)\}\\
&\semantics{\comparacionCaminos{\aPath \neq \aPathb}}_\aGraph = \{v \mid \exists u, w,\, (v,u) \in \semantics{\aPath}_\aGraph,\, (v,w) \in \semantics{\aPathb}_\aGraph,\, D(u) \neq D(w)\} 
\end{alignat*}

We use $\aPath \entoncesCamino \aPathb$ to denote the path expression $\aPathb \pathUnion \pathComplement{\aPath}$, and $\aNodeExpression \entoncesNodo \aNodeExpressionb$ to denote the node expression $\aNodeExpressionb \lor \neg \aNodeExpression$. We also note a label $\aLabel$ as $\down_\aLabel$ in order to easily distinguish the `path' fragment of the expressions. For example, the expression \aLabel [\esDatoIgual{c}]\aLabel\, will be noted as $\down_\aLabel [\esDatoIgual{c}] \down_\aLabel$.
Naturally, the expression $\aPath \cap \aPathb$ can be rewritten as $\overline{\overline{\aPath} \cup \overline{\aPathb}}$ while preserving the semantics. Something similar happens with the operators $\wedge$ and $\vee$ for the case of node expressions using the $\lnot$ operator. We define all these operators in this grammar since further on we will be interested in a fragment of $\Gregxpath$ called $\Gposregxpath$, which has the same grammar except for the $\overline{\aPath}$ and $\lnot \aFormula$ productions. Thus, in \Gposregxpath we will not be able to `simulate' the $\cap$ operator unless it is present in the original \Gregxpath grammar.

We will also denote by $\Gposregxpathnode$ the subset of $\Gposregxpath$ that only contains node expressions.

\vspace{-2mm}
\paragraph{Consistency}
Given a specific database, we want node or path expressions to represent some structural property we expect to find in our data. This kind of \Gcorexpath or \Gregxpath expression works as an \textit{integrity constraint} by defining semantic relations among our data. 
Formally, we define the notion of consistency in the following way:

\begin{definition}[Consistency]
Let $\aGraph$ be a data-graph and $\aRestrictionSet = \aRestrictionSetPaths \union \aRestrictionSetNodes$ a set of restrictions, where $\aRestrictionSetPaths$ consists of path expressions and $\aRestrictionSetNodes$ of node expressions. 
We say that $\aGraph$ is \defstyle{consistent} w.r.t.\ $\aRestrictionSet$, denoted by $\aGraph \models \aRestrictionSet$, if the following conditions hold: \emph{(a)} for all $\aNodeExpression \in \aRestrictionSetNodes$, we have that $\semantics{\aNodeExpression} = V_\aGraph$,
\emph{(b)} for all $\aPath \in \aRestrictionSetPaths$, we have that $\semantics{\aPath} = V_\aGraph \times V_\aGraph$. 
Otherwise, we say that $\aGraph$ is \defstyle{inconsistent} w.r.t.\ $\aRestrictionSet$. 
\end{definition}

In the rest of the paper, we will simply say that $\aGraph$ is (in)consistent whenever the restriction set $\aRestrictionSet$ is clear from the context. 

\begin{example}

 Consider the film database from Figure \ref{figure:movieDB}, where we have nodes representing people from the film industry (such as actors or directors) and others representing movies or documentaries. 

 \begin{figure}[h]

\centering

\begin{tikzpicture}[node distance={25mm}, thick, main/.style = {draw, rectangle}, scale=0.65, every node/.append style={transform shape}]

\node[main] (Hoffman) {Hoffman};
\node[main] (Actor) [left of=Hoffman] {Actor};
\node[main] (Phoenix) [below left of=Actor] {Joaquin Phoenix};
\node[main] (Robbie) [below right of=Hoffman] {Margot Robbie};
\node[main] (Babylon) [above right of=Robbie] {Babylon};
\node[main] (The Master) [above of=Hoffman] {The Master};
\node[main] (film) [right of=Babylon] {Film};
\node[main] (Anderson) [above of=film] {Anderson};
\node[main] (Chazelle) [right of=Anderson] {Chazelle};

\draw[->] (Hoffman) -- (Actor) node[midway, above=0.2pt, sloped]{type};

\draw[->] (Phoenix) -- (Actor) node[midway, above=0.1pt, sloped] {type};

\draw[->] (Robbie) -- (Actor) node[midway, above=0.2pt, sloped] {type};


\draw[->, bend left=50] (Phoenix) edge  node[midway, above=0.1pt, sloped] {acts\_in} (The Master) ;

\draw[->] (Hoffman) -- (The Master) node[midway, sloped, above=0.1pt]{acts\_in};

\draw[->] (The Master) -- (Anderson) node[midway, above=0.1pt] {directed\_by};

\draw[->] (Babylon) -- (Chazelle) node[midway, above=0.1pt, sloped] {directed\_by};

\draw[->] (Robbie) -- (Babylon) node[midway, above=0.1pt, sloped] {acts\_in};

\draw[->] (Babylon) -- (film) node[midway, below=0.1pt] {type};

\draw[->, bend right=10] (The Master) edge node[midway, above=0.1pt, sloped] {type} (film) ;

    
\end{tikzpicture} 
\caption{A film data-graph.}
\label{figure:movieDB}
\end{figure}
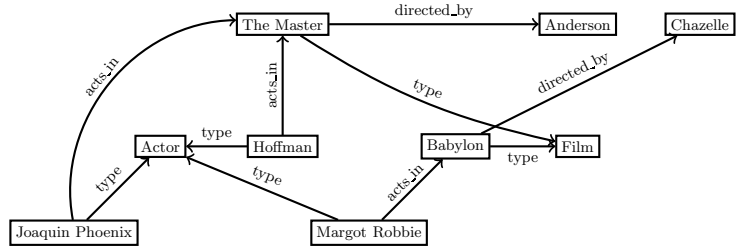
 
 If we want to make a cut from that graph that preserves only actors who have worked with Philip Seymour Hoffman through a film by Paul Thomas Anderson, then we want the following formula to be satisfied:
\begin{equation*}
     \aNodeExpression = \comparacionCaminos{\down_\esLabel{type}[actor^=]}   \entoncesNodo \comparacionCaminos{\down_\esLabel{acts\_in}\comparacionCaminos{\down_\esLabel{directed\_by} [\textit{Anderson}^=]} \down_\esLabel{acts\_in}^- [\textit{Hoffman}^=]}.
\end{equation*}

Notice that $\aNodeExpression$ is not satisfied in the depicted data-graph, since Robbie did not work with Hoffman in a film directed by Anderson, hence we do not have consistency with respect to $\{ \aNodeExpression \}$. Observe that the restriction also applies to Hoffman, thus it is required that he participates in at least one film by Anderson in order to satisfy the constraint.


\end{example}


\vspace{-5mm}
\paragraph{Repairs}
If a graph database $\aGraph$ is inconsistent with respect to a set of restrictions $\aRestrictionSet$ (i.e. there is a path expression or node expression in $\aRestrictionSet$ that is not satisfied), we would like to compute a new graph database $\aGraph'$ consistent with respect to $\aRestrictionSet$ that minimally differs from $\aGraph$. This new database $\aGraph'$ is usually called a \textit{repair} of $\aGraph$ with respect to $\aRestrictionSet$, following some formal definition for the semantics of `minimal difference'.

Here we consider \textit{set repairs}, in which the notion of minimal difference is based on sets of nodes and edges. While we could provide a notion of distance between arbitrary data-graphs via an adequate definition of symmetric difference, it has been the case that the complexity of finding such repairs is quite high, so it is common to consider set repairs where one graph is obtained from the other by only adding or only deleting information
\cite{barcelo2017data,lukasiewicz2013complexity,tenCate:2012}. This gives raise to subset and superset repairs, on which we focus in this work.

We say that a data-graph $\aGraph = (V,L_e,D)$ is a \defstyle{subset} of a data-graph $\aGraph'=(V',L_e',D')$ (written as $\aGraph \subseteq \aGraph'$) if and only if $V \subseteq V'$ and for all $v,v' \in V$ it happens that $L_e(v,v') \subseteq L_e'(v,v')$ and $D(v) = D'(v)$. In this case, we also say that $\aGraph'$ is a \defstyle{superset} of $\aGraph$.

\begin{definition}[Subset  and superset repairs]
Let $\aRestrictionSet$ be a set of restrictions and $\aGraph$ a data-graph. We say that $\aGraph'$ is a \defstyle{subset repair} (resp.\, \defstyle{superset repair}) or $\subseteq$-repair (resp.\, $\supseteq$-repair) of $\aGraph$ if: \emph{(a)} $\aGraph' \models \aRestrictionSet$, \emph{(b)} $\aGraph' \subseteq \aGraph$ (resp.\, $\aGraph' \supseteq \aGraph$), and \emph{(c)} there is no data-graph $\aGraph''$ such that $\aGraph''\models\aRestrictionSet$ and $ \aGraph' \subset \aGraph'' \subseteq \aGraph$ (resp.\, $\aGraph' \supset \aGraph'' \supseteq \aGraph$).
%
%
%
We note the set of subset (resp.\ superset) repairs of $\aGraph$ with respect to $\aRestrictionSet$ as $\subseteq$-$Rep(\aGraph,\aRestrictionSet)$ (resp.\, $\supseteq$-$Rep(\aGraph,\aRestrictionSet)$). 
\end{definition}

\begin{example}
In Example~\ref{figure:movieDB}, by deleting the node with value $\esDato{Margot Robbie}$ we obtain a $\subseteq$-repair of the graph database. 
\end{example}

\vspace{-3mm}
\paragraph{Preferences} Now we introduce the two preference criteria that we will use to induce orderings on the set of repairs.
In the manner done in \cite{bienvenu2014querying} for Description Logic knowledge bases, we provide a notion of {\it weight} over graph databases, which can be translated into preferences via the induced ordering.

\begin{definition}[Weight functions]
    Given a function $w: \Sigma_e \sqcup \Sigma_n \to \N$ (where $\sqcup$ denotes disjoint union), we can extend $w$ to any finite data-graph $G = (V,L_e,D)$ over $\Sigma_e$ and $\Sigma_n$ as
\vspace{-.25cm}
$$\weight(G) = \sum_{x,y \in V} \left(\sum_{z \in L_e(x,y)} \weight(z) \right) + \sum_{x \in V} \weight(D(x)).$$
\end{definition}





When considering a way to select one among various possible subset or superset repairs, one approach is to consider that different edge labels and data values are prioritized differently by being assigned different weights. These weights can be aggregated to obtain a measure of the weight of a whole data-graph, and this aggregated value can then be compared for all the possible repairs to obtain a preferred repair based on the natural ordering of non-negative integers.

\begin{definition}[Weight-based preferences]
Given a weight function $\weight$, we define that $G_1 \lessPreferedInducedWeight G_2$ iff $\weight(G_1) < \weight(G_2)$.

If $G_1 \lessPreferedInducedWeight G_2$, we say that $G_1$ is \defstyle{\weight-preferred} to $G_2$ in the context of superset repairs, while we say that $G_2$ is \defstyle{\weight-preferred} to $G_1$ in the case of subset repairs.
\end{definition}


\begin{example}\label{example:WeightsAndRepairs}
Consider a context where data-graphs represent physical networks, and where edges represent two different quality levels of connection (e.g. varying robustness, resistance to physical attacks) which we call $\downarrow_\esLabel{low}$ and $\downarrow_\esLabel{high}$. 
Let $\aRestrictionSet = \{\aPath_{connected\_dir}, \aPath_{2l\rightarrow good}\}$ be a set of restrictions, where 
\begin{align*}
 &\aPath_{connected\_dir} = \labelComodin^*   \\
 &\aPath_{2l\rightarrow good} = \downarrow_\esLabel{low} \downarrow_\esLabel{low}
\entoncesCamino \downarrow_\esLabel{high} \downarrow_\esLabel{low} \pathUnion \downarrow_\esLabel{low} \downarrow_\esLabel{high} \pathUnion \downarrow_\esLabel{high} \downarrow_\esLabel{high} \pathUnion \downarrow_\esLabel{high} \pathUnion \downarrow_\esLabel{low}.
\end{align*}

$\aPath_{connected\_dir}$ expresses the notion of directed connectivity, and 
$\aPath_{2l\rightarrow good}$  
establishes that if a node can be reached by two low-quality edges, then it is also possible to reach it by a `good' path. That is, it can be reached in either only one step, or in two steps but using at least one high-quality edge.

We could consider a weight function that attempts to represent the costs of building nodes and connections in this network. For example, it could assign a uniform weight to all data values $\weight(\aNode) = 20$, a low cost for low-quality connections $\weight(\downarrow_\esLabel{low}) = 1$, and higher costs for high-quality connections $\weight(\downarrow_\esLabel{high}) = 5$.

Now, given a data-graph $\aGraph$ that does not satisfy the restrictions, a $\weight$-preferred superset repair can be interpreted as the most cost-effective way of making a superset of the network that satisfies the restrictions while minimizing the costs given by $\weight$. For a full example, see Figure~\ref{figures:weightExamplePreferredRepair}.
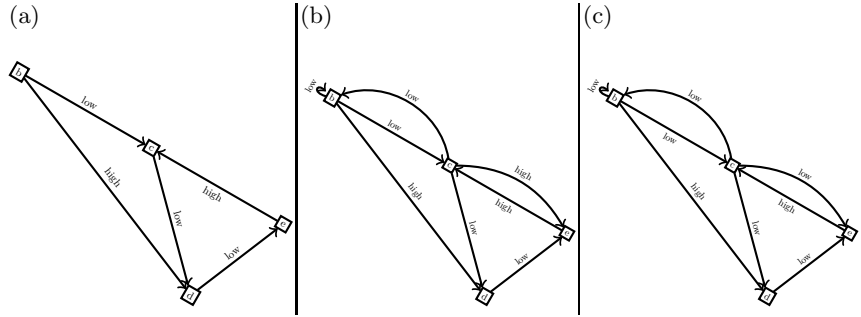
\begin{figure}[h]
	\centering
	\begin{tabular}{ l | l | l }
		(a) & (b) & (c)  \\ & & \\
		\begin{tikzpicture}%
			[rotate=-30,node distance={45mm}, thick, main/.style = {draw, rectangle},  scale=0.45, every node/.append style={transform shape}] 
			\node[main] (1) {c}; 
			\node[main] (2) [below right of=1] {d}; 
			\node[main] (3) [left of=1] {b}; 
			\node[main] (4) [right of=1] {e};
			\draw[->] (3) -- (1) node[midway, above=0.5pt, sloped]{low}; 
			\draw[->] (4) edge node[midway, below=0.5pt, sloped, pos=0.5]{high} (1);
			\draw[->] (1) -- (2) node[midway, above=0.5pt, sloped, pos=0.5]{low};
			\draw[->] (3) -- (2) node[midway, above=0.5pt, sloped, pos=0.5]{high};
			\draw[->] (2) -- (4) node[midway, above=0.5pt, sloped, pos=0.5]{low}; 
		\end{tikzpicture}
		&
		\begin{tikzpicture}%
			[rotate=-30,node distance={45mm}, thick, main/.style = {draw, rectangle},  scale=0.4, every node/.append style={transform shape}] 
			\node[main] (1) {c}; 
			\node[main] (2) [below right of=1] {d}; 
			\node[main] (3) [left of=1] {b}; 
			\node[main] (4) [right of=1] {e};
			\draw[->] (3) -- (1) node[midway, above=0.5pt, sloped]{low};
			\draw[->, loop left] (3) edge node[midway, above=0.5pt, sloped]{low} (3); 
			\draw[->] (4) edge node[midway, below=0.5pt, sloped, pos=0.5]{high} (1);
			\draw[->, bend left = 30] (1) edge node[midway, above=0.5pt, sloped, pos=0.5]{high} (4); 
			\draw[->] (1) -- (2) node[midway, above=0.5pt, sloped, pos=0.5]{low};
			\draw[->] (3) -- (2) node[midway, above=0.5pt, sloped, pos=0.5]{high};
			\draw[->] (2) -- (4) node[midway, above=0.5pt, sloped, pos=0.5]{low};
			\draw[->, bend right = 50] (1) edge node[midway, above=0.5pt, sloped, pos=0.5]{low} (3); 
		\end{tikzpicture}
		&
		\begin{tikzpicture}%
			[rotate=-30,node distance={45mm}, thick, main/.style = {draw, rectangle},  scale=0.4, every node/.append style={transform shape}] 
			\node[main] (1) {c}; 
			\node[main] (2) [below right of=1] {d}; 
			\node[main] (3) [left of=1] {b}; 
			\node[main] (4) [right of=1] {e};
			\draw[->] (3) -- (1) node[midway, below=0.5pt, sloped]{low};
			\draw[->, loop left] (3) edge node[midway, above=0.5pt, sloped]{low} (3); 
			\draw[->] (4) edge node[midway, below=0.5pt, sloped, pos=0.5]{high} (1);
			\draw[->, bend left = 30] (1) edge node[midway, above=0.5pt, sloped, pos=0.5]{low} (4); 
			\draw[->] (1) -- (2) node[midway, above=0.5pt, sloped, pos=0.5]{low};
			\draw[->] (3) -- (2) node[midway, above=0.5pt, sloped, pos=0.5]{high};
			\draw[->] (2) -- (4) node[midway, above=0.5pt, sloped, pos=0.5]{low};
			\draw[->, bend right = 50] (1) edge node[midway, above=0.5pt, sloped, pos=0.5]{low} (3); 
		\end{tikzpicture}
	\end{tabular}
	\caption{(a) A data-graph that does not satisfy $\aPath_{connected\_dir}$ nor $\aPath_{2l\rightarrow good}$ from Example~\ref{example:WeightsAndRepairs}: it is not connected as a directed graph, and the pair $(c,e)$ of nodes is connected via two $\downarrow_\esLabel{low}$ but cannot be connected via a `good' path. (b) A possible $\weight$ superset repair with respect to the example of figure (a); note that the removal of any of the new edges ends up violating $\aRestrictionSet$. The associated weight of this repair is the original weight plus 1 + 1 + 5 (from the added \esLabel{low} and \esLabel{high} edges). (c) A $\weight$-preferred superset repair with respect to the example of figure (a). The associated extra weight of this repair is 3, and it can be proved that there is no other superset repair with a lower weight.}
	\label{figures:weightExamplePreferredRepair}
\end{figure}
\end{example}


 The weight function $\weight$ (over $\Sigma_n$ and $\Sigma_e$) is considered fixed in general, and it should be an `easy' function to compute. That is, given a reasonable encoding for $\Sigma_e \sqcup \Sigma_n$, we expect that $\weight(x) \leq 2^{p(|x|)}$ for every $x \in \Sigma_e \sqcup \Sigma_n$ and some polynomial $p(n)$, and also assume that the result of $\weight(x)$ should be computable in polynomial time over $|x|$, the \textit{size} of $x$. Other kinds of restrictions could be made upon $\weight$ (for example, that $\weight(x) \leq q(|x|)$ for some polynomial $q(n)$) depending on the kind of weight function wanted to be modeled, but we note that, without any restriction, $w$ could even be an uncomputable function. 



\

The second type of preference criteria we study is based on \textit{multisets}. 

\vspace{-1mm}
\begin{definition}[Multisets]
 Given a set $A$, its set of \defstyle{finite multisets} 
 is 
 defined as 
$\multiset{A} =\{ M: A \to \N \mid M(x) \neq 0 \text{ only for a finite number of }x \}.$
Given a strict partial order $(A, <)$, the \defstyle{multiset ordering} 
$(\multiset{A}, \lessMultiset)$ is defined as in {\em\cite{DershowitzManna,huet1980equations}}: $M_1 \lessMultiset M_2$ iff $M_1 \neq M_2$ and for all $x \in A$, if $M_1(x) > M_2(x)$, then there exists some $y \in A$ such that $x < y$ and  $M_1(y) < M_2(y)$.
\end{definition}


 If $(A, <)$ is a strict partial (resp. total) order, then $(\multiset{A}, \lessMultiset)$ is a partial (resp. total) order. If $(A, <)$ is a well-founded order\footnote{I.e. for all $S \subseteq A$, if $S \neq \emptyset$ then there exists $m \in S$ such that $s\not < m$ for every $s \in S$.}, then  we have that $(\multiset{A}, \lessMultiset)$ is also a well-founded order \cite{DershowitzManna}.


\begin{definition}
Given a finite data-graph $G$ over $\Sigma_e$ and $\Sigma_n$, we define its \defstyle{multiset of edges and data values} as the multiset \edgeDataMultiset{G} over $\Sigma_e$ and $\Sigma_n$ such that: 

$$ 
\edgeDataMultiset{G}(x) = \begin{cases}
|\semantics{x}_\aGraph| & x \in \Sigma_e \\
|\semantics{x^=}_\aGraph| & x \in \Sigma_n.  
\end{cases}
$$
 \end{definition}

 Note that all these multisets of edges and data values belong to $\multiset{A}$ with $A = \Sigma_e \sqcup \Sigma_n$.

 \begin{definition}[Multiset-based preferences]\label{definition:MultisetPreferencesFromPartialOrder}
 Let $G_1, G_2$ be two finite data-graphs over $\Sigma_e$ and $\Sigma_n$, and let $<$ be a partial order defined over 
 $A = \Sigma_e \sqcup \Sigma_n$. 
We say that $G_1 \lessPreferedInducedMultiset G_2$ if 
 $\edgeDataMultiset{G_1} \lessMultiset \edgeDataMultiset{G_2} $.

 If $G_1 \lessPreferedInducedMultiset G_2$, we say that $G_1$ is \defstyle{\multisetCriteria-preferred} to $G_2$ in the context of superset repairs, while we say that $G_2$ is \defstyle{\multisetCriteria-preferred} to $G_1$ in the case of subset repairs.
 

 \end{definition}



\begin{example}\label{example:multisetRepair}
    Consider the data-graphs from Figure~\ref{figures:weightExamplePreferredRepair}. Ignoring any possible data value, observe that the multisets corresponding to graphs b) and c) are (with the informal multiset notation): $\{\esLabel{low}, \esLabel{low}, \esLabel{low}, \esLabel{low}, \esLabel{low},\esLabel{high}, \esLabel{high}, \esLabel{high}\}$ and $\{\esLabel{low}, \esLabel{low}, \esLabel{low}, \esLabel{low}, \esLabel{low}, \esLabel{low},\esLabel{high}, \esLabel{high}\}$, respectively. Assuming $\esLabel{low}>\esLabel{high}$, then in this case data-graph b) is \multisetCriteria-preferred to c). 
\end{example}


\section{Preferred repairs}\label{Section:PrefRepairs}

In this section we consider subset and superset repairs and \preferenceCriteria-preferred criteria where  $\preferenceCriteria \in \{\weight, \multisetCriteria\}$, using different subsets of \Gregxpath for $\mathcal{L}$.

We will always consider the weight function $w:\Sigma_e \sqcup \Sigma_n \to \N$ fixed and efficiently computable: given a codification of $x \in \Sigma_e \sqcup \Sigma_n$ of size $n$ the value $w(x)$ is computable in $poly(n)$. This implies that, for any data-graph $\aGraph$, $w(\aGraph)$ is also computable in $poly(|\aGraph|)$.

In the same manner, when considering multiset-preferred repairs, we will assume that the order $<$ defined over $\Sigma_e \sqcup \Sigma_n$ can be computed efficiently: given $x,y \in \Sigma_e \sqcup \Sigma_n$ it is possible to decide if $x<y$ in polynomial time on the representation on both $x$ and $y$. Thus, we can decide whether $G_1 <_{Gmset} G_2$ in $poly(|G_1| + |G_2|)$. Furthermore, we assume that the order $<$ is well founded. This implies that there are no infinite descending chains $\aData_1 > \aData_2 > \ldots$

Note that the notions of weight-based and multiset-based preferences induce an ordering over finite data-graphs, which does not admit infinite descending chains. Hence, if a superset repair exists, then there is also a (weight-based or multiset-based) preferred repair. On the other hand, the number of subsets of a given data-graph is finite, so if a repair exists, there must necessarily be a preferred repair.

The complexity 
of the problems when the set of expressions $\aRestrictionSet$ is fixed is commonly denominated \textit{data complexity}. Most lower bounds we derive apply to this case.

\subsection{Preferred Subset Repairs}

For the case of $\subseteq$-repairs, it was proved in \cite{Abriola-etal-2023} that deciding if there exists a non-trivial repair (i.e. different from the $\emptyset$ data-graph) is an \textsc{NP-complete} problem for a fixed set $\aRestrictionSet$ of $\Gposregxpath$ expressions, and that the problem is tractable if we only allow node expressions from \Gposregxpath as $\mathcal{L}$. 
Observe that:

\begin{proposition}\label{prop:reduction-subset}
The problem of deciding if $\aGraph$ has a non-trivial $\subseteq$-repair with respect to $\aRestrictionSet$ can be reduced to the problem of deciding whether $\aGraph$ has a non-trivial $\preferenceCriteria$-preferred $\subseteq$-repair with respect to $\aRestrictionSet$.
\end{proposition}

Given a fixed weight function $\weight$ (resp.\ an ordering $<$), the existence of a $\subseteq$-repair $\aGraph'$ of $\aGraph$ with respect to $\aRestrictionSet$ implies the existence of a preferred $\subseteq$-repair for $\aGraph$ (and vice versa). Therefore, it follows directly from Proposition~\ref{prop:reduction-subset} and the aforementioned results in \cite{Abriola-etal-2023} that:

\begin{theorem}\label{teo:subset-w-preferred}
The problem of deciding if there exists a non-trivial \preferenceCriteria-preferred $\subseteq$-repair for a given data-graph $\aGraph$ and a set of $\Gposregxpath$ expressions $\aRestrictionSet$ is \textsc{NP-complete} for a fixed set of $\Gposregxpath$ path expressions.
\end{theorem}

When $\mathcal{L} \subseteq \Gposregxpathnode$, a subset repair can be computed in polynomial time and, furthermore, it is unique \cite{Abriola-etal-2023}. Therefore, it must be the preferred one:

\begin{theorem}
    Given a data-graph $\aGraph$, a set of $\Gposregxpath$ expressions $\aRestrictionSet$, and a preference criteria $\preferenceCriteria$, there exists an algorithm that computes the $\preferenceCriteria$-preferred $\subseteq$-repair of $\aGraph$ with respect to $\aRestrictionSet$. 
\end{theorem}

\subsection{Preferred Superset Repairs}


There is a restriction set $\aRestrictionSet$ containing only $\Gposregxpathnode$ expressions such that computing $\weight$-preferred $\supseteq$-repair is already intractable:

\begin{theorem}\label{teo:supset-weight-hard}
    Given a data-graph $\aGraph$, a set of $\Gposregxpathnode$ expressions $\aRestrictionSet$ and a natural number $K$, let $\Pi_\weight$ be the problem of deciding if there exists a $\weight$-preferred $\supseteq$-repair of $\aGraph$ with respect to $\aRestrictionSet$ whose weight is bounded by $K$.
    Then, there exists a set of positive node expressions $\mathcal{R}$ and a weight function $\weight$ such that the problem is \textsc{NP-complete}.
\end{theorem}

\begin{proof}
    The problem is in \textsc{NP} in general: if there exists a repair of $\aGraph$ with respect to $\aRestrictionSet$, due to \cite[Theorem~24]{Abriola-etal-2023}, then there is one in `standard form'. The size of this repair is bounded by $poly(|\aGraph| + |\aRestrictionSet|)$, and by inspecting the proof of the theorem and considering that $\weight$ is a non-negative function, it can be shown that the preferred repair always has this `standard form'. Then, a positive certificate consists of a data-graph $\aGraph'$ such that $|\aGraph'| \leq poly(|\aGraph| + |\aRestrictionSet|)$, $\aGraph'\models \aRestrictionSet$, $\aGraph \subseteq \aGraph'$ and $w(\aGraph') \leq K$.

    For the hardness, we reduce 3\textsc{-SAT} to $\Pi_\weight$, with $\aRestrictionSet$ fixed. For the reduction, we consider $\Sigma_e = \{\esLabel{value\_of}, \esLabel{appears\_in}, \esLabel{appears\_negated\_in}\}$ and $\Sigma_n = \{clause, var, \top,\bot\}$. We define the weight function $w$ as $w(c) = 2$ for $c \in \Sigma_n \cup \{\esLabel{appears\_in}, \esLabel{appears\_negated\_in}\}$ and $w(\esLabel{value\_of)} = 1$.

Given a 3\textsc{-CNF} formula $\phi$ with $n$ variables and $m$ clauses we build a data-graph $\aGraph$, a set of $\Gposregxpathnode$ node expressions $\aRestrictionSet$ and define a number $K$ such that $\phi$ is satisfiable if and only if $\aGraph$ has a superset repair $\aGraph'$ with respect to $\aRestrictionSet$ such that $w(\aGraph') \leq K$.

We define the graph as $G=(V_\aGraph, L_\aGraph, D_\aGraph)$ where:
\begin{align*}
    &V_\aGraph = \{x_i \mid 1 \leq i \leq n\} \cup \{c_j \mid 1\leq j \leq m\} \cup \{\bot, \top\} \\
    &L_\aGraph(x_i,c_j) =  \begin{cases} 
      \{\esLabel{appears\_in}, \esLabel{appears\_negated\_in}\}\text{ if } x_i \text{ and } \lnot x_i \text{ appear in }c_j \\
      \{\esLabel{appears\_in}\}\text{ if only } x_i \text{ appears in }c_j \\
      \{\esLabel{appears\_negated\_in}\}\text{ if only }\lnot x_i \text{ appears in }c_j \\
      \emptyset \text{ otherwise} 
   \end{cases} \\
   &L_\aGraph(v,w) = \emptyset \text{ for every other pair }v,w \in V_\aGraph \\
   &D_\aGraph(x) = x \text{ for } x\in\{\bot,\top\} \\
   &D_\aGraph(x_i) = var \text{ for }1\leq i \leq n  \\
   &D_\aGraph(c_j) = clause \text{ for }1 \leq j \leq m.
\end{align*}

    
    
   
   
   
   
   

The structure of $\phi$ is codified in the \esLabel{appears} edges. We define $K=w(\aGraph) + n$, and we want any superset repair of $\aGraph$ with respect to $\aRestrictionSet$ with weight $K$ to codify an assignment of the `node' variables by using the edges \esLabel{value\_of}. In order to do this we define the \Gposregxpathnode expressions
\begin{align}
    \aFormulab_1 &= \comparacionCaminos{[var^{\neq}] \cup \down_\esLabel{value\_of} [\top] \cup \down_\esLabel{value\_of} [\bot]}\\
    \aFormulab_2 &= \comparacionCaminos{ [clause^{\neq}] \cup \down_\esLabel{appears\_in}^- \down_\esLabel{value\_of} [\top] \cup \down_\esLabel{appears\_negated\_in}^- \down_\esLabel{value\_of} [\bot]}. 
\end{align}

The expression $\psi_1$ forces every variable node to have a \esLabel{value} edge directed to a boolean node, while the expression $\psi_2$ forces every clause to be `satisfied' in any repair of $\aGraph$. Therefore, we define $\aRestrictionSet = \{ \aFormulab_1, \aFormulab_2\}$. Now we show that $\phi$ is satisfiable if and only if $\aGraph$ has a superset repair with respect to $\aRestrictionSet$ with weight at most $w(\aGraph) + n$.

$\implies)$ Let $f$ be a valuation on the variables of $\phi$ that evaluates $\phi$ to true. We then `add' edges to $\aGraph$ in the following way: if $f(x_i) = \top$ we add the edge $(x_i,\esLabel{value\_of},\top)$, and otherwise we add $(x_i,\esLabel{value\_of},\bot)$. This graph satisfies both expressions from $\aRestrictionSet$ and has cost $w(\aGraph) + n$, since $w(\esLabel{value\_of}) = 1$.

$\impliedby)$ Let $\aGraph'$ be a superset repair of $\aGraph$ with respect to $\aRestrictionSet$ with weight at most $w(\aGraph) + n$. Since every superset repair of $\aGraph$ has to add at least $n$ \esLabel{value} edges and they cost 1 unit each we know that $\aGraph'$ has to be exactly the original graph $\aGraph$ plus $n$ \esLabel{value} edges (one for each node variable). Then we can define a valuation of the variables of $\phi$ by using this edges: if the edge $(x_i,\esLabel{value\_of}, \top)$ is present in $\aGraph'$ we define $f(x_i) = \top$, and otherwise $f(x_i) = \bot$. Since $\aFormulab_2$ is satisfied in $\aGraph'$ this assignment must satisfy $\phi$. \qed
\end{proof}

The hardness of $\Pi_\weight$ implies that, unless $\textsc{P}=\textsc{NP}$ there is no algorithm to compute $\weight$-preferred $\supseteq$-repairs for fixed sets of $\Gposregxpathnode$ expressions.

Furthermore, computing a multiset preferred repair is also a hard problem for simple $\Gposregxpathnode$ expressions:

\begin{theorem}\label{teo:supset-multiset-hard}

Given a data-graph $\aGraph$, a set of $\Gposregxpathnode$ expressions $\aRestrictionSet$, an edge label $\aLabel \in \Sigma_e$, and a natural number $K$, let $\Pi_\multisetCriteria$ be the problem of deciding if $\aGraph$ has a $\multisetCriteria$-preferred $\supseteq$-repair $\aGraph'$ with respect to $\aRestrictionSet$ such that $\aGraph'$ has at most $K$ edges with label $\aLabel$. Then, there exists a set of positive node expressions $\aRestrictionSet$ and a well-ordering $<$ such that the problem is \textsc{NP-complete}.

\end{theorem}

\begin{proof}
    The proof is analogous to the one of Theorem~\ref{teo:supset-weight-hard}, considering the order $\esLabel{value\_of} < \esLabel{appears\_in} < \esLabel{appears\_negated\_in} < clause < var < \top < \bot$ and $K=n$. \qed
\end{proof}

It follows that the hardness of $\Pi_\multisetCriteria$ implies the hardness of computing $\multisetCriteria$-preferred $\supseteq$-repairs for fixed sets of $\Gposregxpathnode$ expressions.
 
\begin{remark}
In the case of $\supseteq$-repairs, it was proved in \cite{Abriola-etal-2023} that deciding if there exists at least one repair is an undecidable problem for a fixed set of $\Gregxpath$ expressions, and \textsc{NP-hard} if $\Gposregxpath \subseteq \mathcal{L}$. Meanwhile, it was shown that if $\mathcal{L}\subseteq \Gposregxpath$ and the restriction set is fixed, or rather if $\mathcal{L} \subseteq \Gposregxpathnode$, then there are polynomial-time algorithms to compute a superset repair. Theorems \ref{teo:supset-weight-hard} and \ref{teo:supset-multiset-hard} show that these tractable cases become intractable when considering preferred repairs.
\end{remark}

\section{Conclusions}\label{Section:Conclusions}
In this work, we analyze preferred repairing for data-graphs. We specifically focus on the problem of deciding if a data-graph $\aGraph$ has a non-trivial preferred repair under two different data-aware preference criteria, one based on weights and the other based on multiset orderings. We showed that in some cases, these criteria do not make the repair decision problem harder than the version lacking preferences. 

Some questions in this context remain open, such as that of finding refined tractable versions of the problem that might be based on real-world applications. Moreover, the precise complexity class for superset repair is still unknown, since we only show a lower bound for this problem. Alternative definitions of types of repairs for data-graphs~\cite{barcelo2017data}, like those based on symmetric-difference~\cite{tenCate:2012}, are worth studying in the preference-based setting. 
Finally, it would be interesting to study more general families of criteria such as the ones proposed in \cite{staworko2012prioritized}, and analyze whether the complexity of the decision problems changes for data-graphs.

\bibliographystyle{abbrv}
\bibliography{biblio}

%
%
%
%

\end{document}